\documentclass[aps,prl,twocolumn]{revtex4-1}

\usepackage{amsmath}
\usepackage{amsthm}

\newtheorem{theorem}{Theorem}[section]

\usepackage{bbm,graphicx,bm,hyperref}   
\usepackage[caption=false]{subfig}

\makeatletter
\newcommand*{\defeq}{\mathrel{\rlap{%
                     \raisebox{0.3ex}{$\m@th\cdot$}}%
                     \raisebox{-0.3ex}{$\m@th\cdot$}}%
                     =}
\makeatother

\def\ii{{\rm i}}
\def\sx{\sigma^{\rm x}}
\def\sy{\sigma^{\rm y}}
\def\sz{\sigma^{\rm z}}
\def\tr#1{{\rm tr}(#1)}
\def\1{\mathbbm{1}}
\def\cL{{\cal L}}
\def\cLH{{\cal L}_{\rm H}}
\def\cLd{{\cal L}_{\rm dis}}
\def\ket#1{{| #1 \rangle}}
\def\bra#1{{\langle #1 |}}
\def\braket#1#2{{\langle #1 | #2 \rangle}}

\def\tit#1{{\em #1},}

\begin{document}

\title{Dissipative remote-state preparation in an interacting medium}

\author{Marko \v Znidari\v c}
\affiliation{Physics Department, Faculty of Mathematics and Physics, University of Ljubljana, Ljubljana, Slovenia}

\date{\today}

\begin{abstract}
Standard quantum state preparation methods work by preparing a required state locally and then distributing it to a distant location by a free-space propagation. We instead study procedures of preparing a target state at a remote location in the presence of an interacting background medium on which no control is required, manipulating only local dissipation. In mathematical terms, we characterize a set of reduced steady states stabilizable by local dissipation. An explicit local method is proposed by which one can construct a wanted one-site reduced steady state at an arbitrary remote site in a lattice of any size and geometry. In the chain geometry we also prove uniqueness of such a steady state. We demonstrate that the convergence time to fixed precision is smaller than the inverse gap, and we study robustness of the scheme in different medium interactions.
\end{abstract}

\pacs{03.65.Yz, 03.67.Hk, 42.50.Dv}



\maketitle

{\em Introduction.--}
Preparation of quantum states is a fundamental prerequisite for quantum technologies~\cite{Nielsen}, e.g., in quantum teleportation~\cite{tele} or quantum computation~\cite{divincenzo}. Frequently, these states are needed at different spatial locations and one has to solve a problem of preparing a given state at a remote place by using only local resources that are spatially separated from the remote location. Because quantum resources needed to prepare a given quantum state are usually involved and expensive, a standard approach is to have a dedicated device that produces states locally, which are then sent through free space to a required location. In the present work we address and solve the question of how to achieve the same if the medium through which one has to ``send'' a state is interacting. One can envisage this interaction to be due to a non-negligible fundamental interacting background, or, e.g., because the whole setting is embedded in a solid-state environment where interactions are ubiquitous, a situation of importance in quantum computation.

We are going to study a concrete setting consisting of a lattice system described by a Markovian master equation of the Lindblad type~\cite{Lindblad}, being within experimental realm~\cite{Barreiro:11,Krauter:11,Lin:13}. An interacting medium is described by a fixed local Hamiltonian, while the operations that one is allowed to make consist of an arbitrary Lindblad evolution on a single site. After a long time an initial state converges to a steady state (SS), and we are interested in a reduced SS on a given remote target site; see also Fig.~\ref{fig:chain}. We want to characterize a set of reduced states stabilizable by local dissipation (also called stabilizable states, or reachable states). 

Existing procedures of transporting a given state to a target location, like doing swap operations, or using quantum wires~\cite{Bose:03,Christandl:04}, all require some control over an interacting medium. In our method we can do without such control. Characterizing the power of open-system~\cite{Breuer} evolution, for instance, the set of reachable states and the controllability of a master equation~\cite{Altafini:03,Ticozzi:08}, has received a lot of attention recently, in particular the optimality of time required to transform a given initial state to a given target state~\cite{Sugny:07,Lapert:10,Giovannetti:13}. Allowing any transformations, one can show that Lindblad equations are in fact a universal resource~\cite{Verstraete:09}. Several other general results are also known, for instance, conditions under which a given pure state can be a SS~\cite{pureSS,Kraus:08}, see also Ref.~\cite{Popkov15}. Having control over unitary evolution allows one to decrease, or even remove, detrimental effects of dissipation~\cite{Recht,Sauer:13}. Frequently, though, we only have limited control and therefore a pressing problem is to characterize the power of constrained resources. In such case there is less symmetry, the problem is more difficult, with only few results available. An important constraint is the locality of the interactions, studied for pure SSs in Ref.~\cite{Ticozzi:12}, for translationally invariant states in Ref.~\cite{JMP:14}, and for frustration-free states in Ref.~\cite{Johnson:15}. It has also been shown that local dissipation limits the lowest attainable temperature~\cite{geometry15}. 

{\em The setting.--} The Lindblad equation is~\cite{Lindblad}
\begin{eqnarray}
\frac{{\rm d}\rho}{{\rm d}t}&=&{\cal L}(\rho)=\ii[\rho,H]+\cLd(\rho),
\label{eq:Lin}
\end{eqnarray}
where $\cLd(\rho)=\sum_k 2L_k \rho L_k^\dagger-\rho L_k^\dagger L_k-L_k^\dagger L_k \rho$ is a dissipator that depends on a set of traceless Lindblad operators $L_k$.
\begin{figure}[!t]
\centering \includegraphics[width=2.3in]{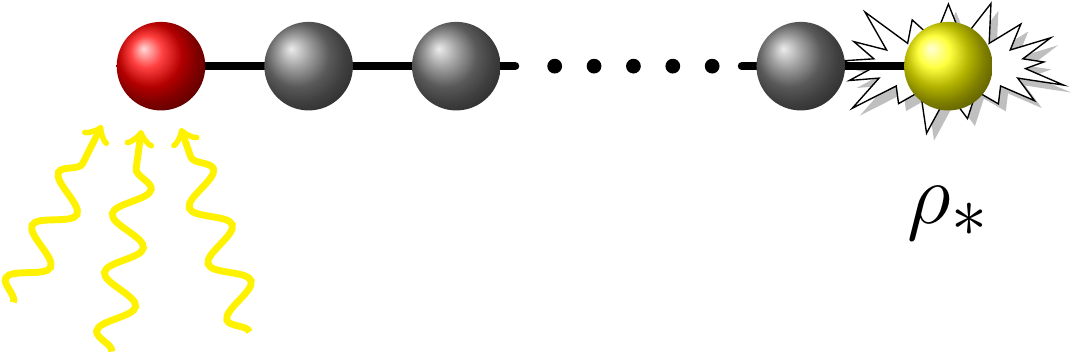}
\caption{(Color online) Remote preparation of states: acting only on the first spin (the red ball), we want to prepare a given target state $\rho_*$ on the last spin (the yellow ball).}
\label{fig:chain}
\end{figure}
After a long time the solution of the Lindblad equation converges to a SS $\rho_\infty=\lim_{t \to \infty} {\rm e}^{\cL t}\rho(0)$, and we are interested in a reduced SS on a given target site $k$, $\rho_k={\rm tr}_{j \neq k}(\rho_\infty)$. We want to characterize the set of $\rho_k$ reachable by controlling only one-site dissipation, keeping $H$ fixed, as well as find a concrete procedure achieving a given $\rho_k=\rho_*$. The following theorem about SSs of permutation Hamiltonians under local one-site dissipation will be of great help.
\begin{theorem}
Let us have a lattice of $n$ sites (each having finite dimension $d$), described by local Lindblad (\ref{eq:Lin}) generator $\cL_m$ acting nontrivially only on the site $m$, and
\begin{equation}
H_{\rm P}=\sum_{j,k} P_{j,k},
\label{eq:P}
\end{equation}
where $P_{j,k}$ is a permutation operator between two sites (acting as $P_{j,k} \ket{\alpha}_j\ket{\beta}_{k}=\ket{\beta}_j\ket{\alpha}_{k}$), and the sum running over an arbitrary set of connections (not necessarily nearest neighbor). Denoting by $\rho_*$ a single-site SS of $\cL_m$, i.e., $\cL_m \rho_*=0$, the SS on the whole lattice is then a product state $\rho_\infty:=\rho_*^{\otimes n}$, $\cL \rho_\infty=0$, where $\cL:=\cL_m\otimes 1+\cLH$ and $\cLH\rho:=\ii [\rho,H_{\rm P}]$. In a one-dimensional chain (with only the nearest neighbor coupling $P_{j,j+1}$) with $\cL_m$ on the edge ($m=1$ or $m=n$), the above $\rho_\infty$ is a unique SS of $\cL$ if and only if $\rho_*$ is a unique SS of $\cL_m$.
\label{Th1}
\end{theorem}
\begin{proof}
The first part is trivial: $\rho_\infty$ is invariant to any permutation, $P_{j,k} \rho_\infty P_{j,k}=\rho_\infty$, and thus $\cLH\rho_\infty=0$. At the same time we also have $(\cL_m\otimes \1) \rho_\infty=0$ because the reduced state of $\rho_\infty$ on the $m$-th site is $\rho_*$. Regarding the uniqueness, it is clear that if $\rho_\infty$ is a unique SS of $\cL$, then $\rho_*$ must be a unique SS of $\cL_m$. For the other direction of the proof, we use the fact that the SS is unique if and only if Lindblad operators, their adjoints, and $H$, span under multiplication and addition the whole operator space~\cite{Evans:77}. If $\rho_*$ is a unique SS of $\cL_m$ (defined in terms of the local Lindblad operators $L_m^{(j)}$ and $H_m$), we know that the set $\{ L_m^{(j)},L_m^{(j)^\dagger},H_m \}$ spans the local operator space at site $m$. All operators at other chain sites can be constructed by the following recursive mapping, $\sum_{i=1}^d \ket{i}\bra{j}_r H_{\rm P} \ket{k}\bra{i}_r = \ket{k}\bra{j}_{r+1}+\ket{k}\bra{j}_{r-1}$, holding for $k\neq j$ (if $r$ is on the edge, the rhs is without one of the terms). Starting from the edge site $m$, we can construct all off-diagonal operators at the neighboring site (and all diagonal ones by products of the off-diagonal). Recursively repeating the procedure we generate the whole basis, progressing from one edge to the other.
\end{proof}
The above SS $\rho_\infty$ is unique also if dissipation acts on any chain site eother than the middle one for an odd $n$ ($m=(n+1)/2$). Potential degeneracy of the SS on other lattices can be removed by placing $\cL_m$ at several sites. Such a $\rho_\infty$ is an example of a frustration-free SS~\cite{Johnson:15}. Hamiltonians treated in the above theorem are in general called ${\rm SU}(d)$ Heisenberg models (chains), important examples being the standard isotropic Heisenberg chain for $d=2$ (where one has $H=\sum_{j=1}^{n-1} \sx_j \sx_{j+1}+\sy_j \sy_{j+1}+\sz_j \sz_{j+1}=\sum_{j=1}^{n-1}[2P_{j,j+1}-\mathbbm{1}]$), or the spin-orbital model~\cite{SO} having $d=4$, i.e., a system with a local two-qubit space. Theorem \ref{Th1} completely answers the question of SSs under strictly local Lindblad dissipation in such systems. Steady states are rather simple from a complexity point of view -- they are simple product states -- however, for our purpose they are just what we need.

{\em Preparation of remote states.--} 
Let us consider a chain lattice composed of $n$ sites, with each site having the dimension $d=2$ (everything we present works for any finite $d$). We would like to prepare an arbitrary target qubit state $\rho_*$ at the far end of our chain (at site $j=n$) by doing operations only on the first site ($j=1$); see Fig.~\ref{fig:chain} . Theorem \ref{Th1} tells us how to proceed: choose a one-site Lindbladian $\cL_1$ that has the wanted $\rho_*$ for the unique SS, and the Heisenberg Hamiltonian. Time evolution by $\cL=\cL_1+\cLH$ then results in $\rho(t)={\rm e}^{\cL t}\rho(0)$, which after a long time converges to the wanted state,
\begin{equation}
\lim_{t \to \infty} {\rm e}^{\cL t}\rho(0)=\rho_*\otimes\rho_* \otimes \cdots \otimes \rho_*,\qquad \forall \rho(0).
\end{equation}
Our procedure is different than the unitary state transfer with quantum wires where a special $H$ is used to gradually transfer a state from one end to the other~\cite{Bose:03,Christandl:04,Bayat:14}. There, at least some control over the wire is required, be it preparation of a special initial state, see though Ref.~\cite{Kim:08}, and/or, e.g., extra engineered magnetic fields. In our scheme no control over the interacting medium is required~\cite{footSwap}: for different $\rho_*$'s we are only adjusting local dissipation at the site $j=1$, while $H$ is held fixed, thereby evolving the system in such a way that the final reduced state at site $j=n$ is $\rho_*$. Also, our procedure is stabilization and not transfer, and is, as such, inherently more robust. It works for any initial state and any sufficiently long time -- we don't have to use a specific initial state or stop at a special time~\cite{Kim:08}. We also note that $\rho(t)$ is, in general, not factorizable at intermediate times, even when starting with a product initial state, and therefore the dynamics cannot be described by a mean-field approximation, like, e.g., in Ref.~\cite{MF}. Regarding the choice of $\cL_1$, there is still a certain freedom as there exist different $\cL_1$'s having the same SS. Several explicit constructions~\cite{footL} are known that use different number of Lindblad operators, e.g., just one Lindblad operator~\cite{Baum08}, or $\log_2{d}$ (for pure states) in Ref.~\cite{Kraus:08}, or a maximal number of $d^2$ Lindblad operators in Ref.~\cite{JSTAT09}. In practice it is important not just that we can prepare an arbitrary state but also how fast and robust the preparation procedure is. We shall study these questions in the rest of the paper. 

{\em Convergence time.--} Convergence time to a stationary state is in general dictated by a spectral gap $g$ of $\cL$. The spectral gap is $g=-\operatorname{\mathbbm{R}e}(\lambda_1)$, where $\lambda_1$ is the eigenvalue of $\cL$ with the largest nonzero real part. Any initial state $\rho(0)$ converges to a unique SS $\rho_\infty$ within a time $\tau$ proportional to the inverse gap, $\tau \sim 1/g$. On general grounds one can argue~\cite{Gaps} that for local dissipation -- our remote-state preparation scheme is an example -- the convergence time must grow at least linearly with the system size, $\tau \sim n$. It has been found~\cite{Medvedyeva14}~\cite{Gaps} though that in integrable systems one typically finds scaling $\tau \sim n^3$. Note that permutation Hamiltonians (\ref{eq:P}) are solvable by Bethe ansatz.
\begin{figure}[!t]
\centering \includegraphics[width=1.61in]{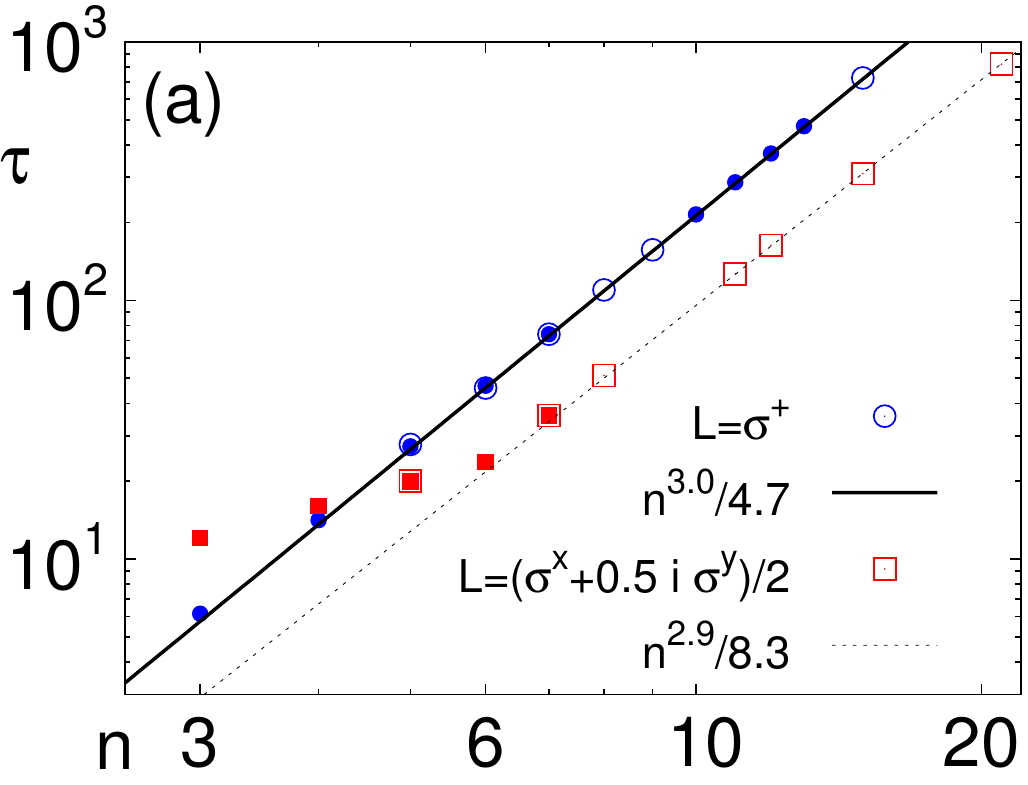} \includegraphics[width=1.69in]{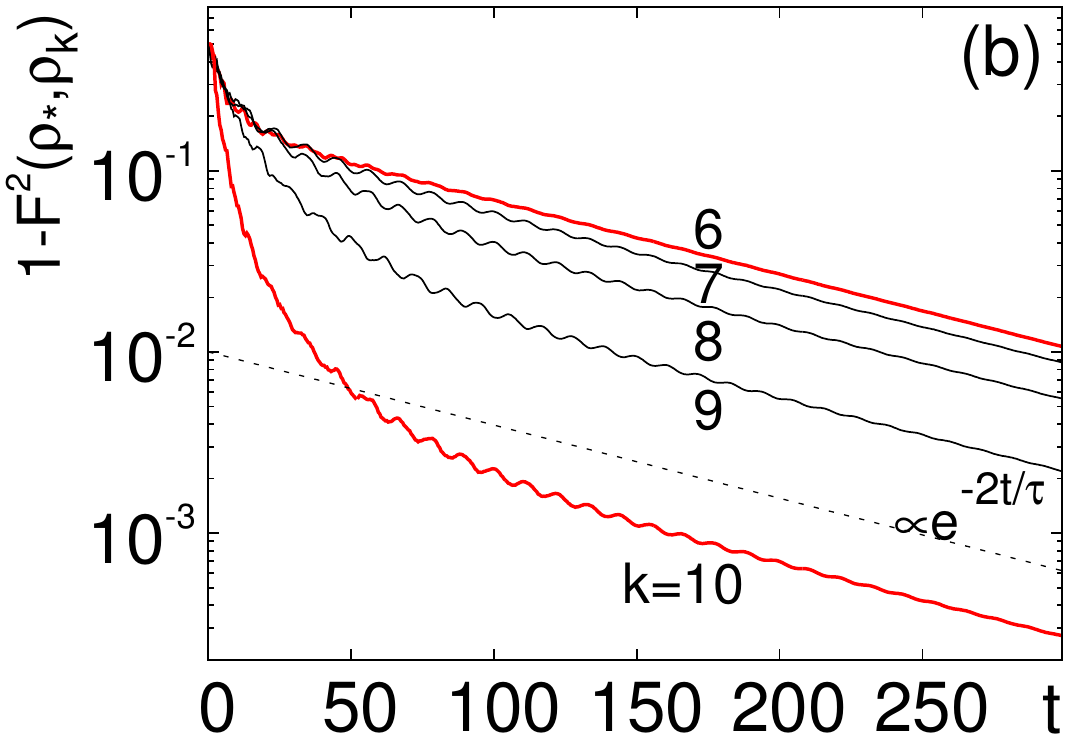} 
\caption{(Color online) (a) Convergence time $\tau=1/g$ for remote-state preparation with the Heisenberg model grows asymptotically as $\tau \sim n^3$. Results for two different target states are shown, mixed $\rho_*=\frac{1}{2}\1+\frac{2}{5}\sz$ (the red squares) and pure $\rho_*=\ket{0}\bra{0}$ (the blue circles). Full symbols are obtained by exact diagonalization, empty by open-system tDMRG adaptation~\cite{NJP10}. (b) Time dependence of fidelity between the target state $\rho_*$ and the reduced state $\rho_k(t)={\rm tr}_{j \neq k}(e^{\cL t}\rho(0))$ at site $k$, all for $n=10$ and $L=\sigma_1^+$. The dashed line is an exponential with the $\tau \approx 215$ read from frame (a).}
\label{fig:convtime}
\end{figure}
As an initial state $\rho(0)$ for our numerical demonstration we use a product pure state with alternating $\ket{0}\pm \ket{1}$ at even/odd sites (similar results are obtained for other choices). In Fig.~\ref{fig:convtime}(a) we can see that for our protocol the Liouvillian gap indeed scales as $\sim n^{-3}$ irrespective of the choice of the target state $\rho_*$. It can happen, though, that the gap is not the whole story and that particular (important) observables converge on a shorter time scale~\cite{JSTAT11}. In addition, the decay in the thermodynamic limit can be different than a simple exponential decay~\cite{Cai:13,Medvedyeva14} (which happens for an isolated $\lambda_1$). With that in mind we also calculated how fast the reduced state at a particular site $k$ approaches its asymptotic SS value $\rho_*$. As a measure of convergence we use quantum fidelity~\cite{Nielsen}, defined as $F(\rho,\sigma)=\tr{\sqrt{\sqrt{\rho}\sigma\sqrt{\rho}}}$. For pure states it simplifies to $F=|\braket{\psi}{\varphi}|$. In Fig.~\ref{fig:convtime}(b) we see that, even though the asymptotic decay is given by the gap, $1-F^2 \sim \exp{(-2tg)}$, the fidelity behaves quite differently at different sites. In particular, the asymptotic exponential decay with time constant $\tau=1/g$ kicks in only after an initial nonexponential decay, duration of which is longer the farther away we are from the middle of the chain ($F$ is approximately the same at sites symmetric with respect to the middle of the chain). At the last (and the first) site the convergence to our target state $\rho_*$ is the fastest (the red line for $k=10$ in Fig.~\ref{fig:convtime}(b)). Compared to state transfer procedures~\cite{Bose:03,Christandl:04}, the state $\rho_*$ does not gradually travel through the chain, instead, the convergence is the fastest at the far-end target site. What is more, the amplitude of the transient initial decay also increases with an increasing $n$. To demonstrate that, we show in Fig.~\ref{fig:scal}(a) the scaling of fidelity at the middle and the last site for different system sizes $n$. We can see that for large times one has a scaling form
\begin{equation}
1-F^2 \asymp \frac{1}{n^\nu}f(t/n^3),
\end{equation}
with some scaling function $f(x)$ that approaches an exponential for a large $x$. We note that, while the shape of the scaling function might depend on a particular choice of the initial state and ${\cal L}_1$, the presented scaling is generic. Interesting is a nontrivial prefactor $1/n^\nu$, with $\nu \approx 0.8$ for the middle site $k=n/2$, and $\nu \approx 2.8$ for the far-end site at $k=n$~\cite{foot4}. As a consequence, the error $1-F^2$ at a fixed time that scales $\sim n^3$ decreases with $n$ as $\sim 1/n^\nu$. This means that the required time to reach a fixed error grows with $n$ slower than $\sim n^3$, see Fig.~\ref{fig:scal}(b). While it is hard to conclude about the exact value of the asymptotic scaling, the convergence time at which a fixed precision $1-F^2$ is reached is closer to $t \sim n^2$ than to $t \sim 1/g=n^3$. This is rather intriguing and has to do with the clustering of eigenvalues around $0$ and the structure of decay eigenmodes.
\begin{figure}[t!]
\centering \includegraphics[width=1.85in]{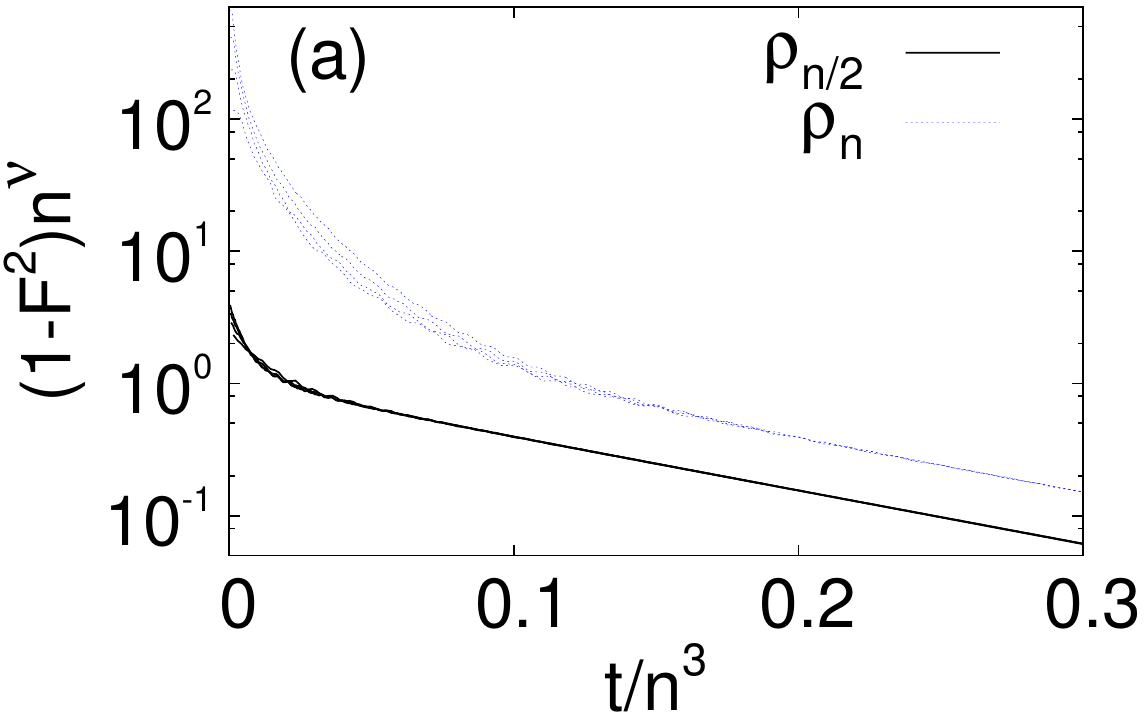} \includegraphics[width=1.5in]{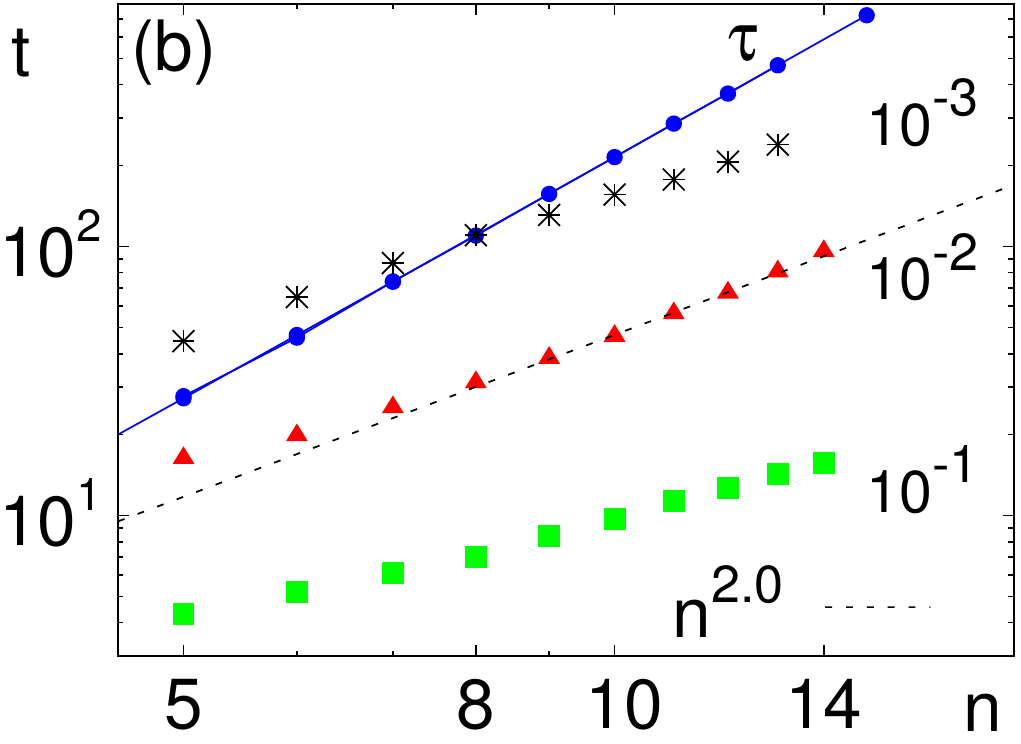}
\caption{(Color online) (a) Fidelity scales as $1-F^2 \asymp \frac{1}{n^\nu}f(t/n^3)$, where the scaling exponent is $\nu \approx 0.8$ for the middle of the chain (the black curves) and $\nu \approx 2.8$ for the end site (the blue/dotted curves). Data is shown for $L=\sigma_1^+$ and $n=7,9,11,13$ (the curves are almost overlapping). (b) Time at which $1-F^2(\rho_*,\rho_n)$ reaches the value $10^{-1}, 10^{-2}, 10^{-3}$ (squares, triangles, stars) grows slower than the inverse gap $\tau \sim n^3$ [the blue circles, data from Fig.~\ref{fig:convtime}(a)], $L=\sigma^+_1$.}
\label{fig:scal}
\end{figure}

{\em Choice of Hamiltonian.--} We next study how different choices of the Hamiltonian influence our remote-state preparation ability. That is, we want to understand whether with other choices of $H$ one can also prepare an arbitrary $\rho_*$ just by varying $\cL_1$. In full generality this is a very difficult question so we will limit our discussion to two important cases. First is a general theorem showing that for a certain type of $H$ only a limited fraction of states can be reached. Second is a full characterization of the set of reachable states for an $XXZ$ type Hamiltonian on $n=2$ qubits, a situation of perhaps the most immediate experimental relevance.

The following theorem limits the set of one-qubit reduced stabilizable states for bipartite systems that have a separable coupling between the target site (subsystem index $n$) and the rest (subsystem ${\rm A}$).
\begin{theorem}
Let us have a master equation with a general Lindblad superoperator $\cL_{\rm A}\otimes \1_n$ (containing arbitrary dissipation as well as a Hamiltonian) and a product coupling Hamiltonian $H=\sz_{n-1}\otimes \sz_{n}$ between one of the spins in ${\rm A}$ and the $n$-th spin. Then the SS ${\rm tr}_{\rm A}{(\rho_\infty)}$ is always diagonal in the eigenbasis of $\sz_{n}$.
\label{Th2}
\end{theorem}
\begin{proof}
The Liouvillian is invariant to rotations around the $\sz_n$ axis, so we can write a separate SS equation $\cL(\rho)=0$ for each subspace $\ket{i}\bra{j}_n$. Taking an off-diagonal SS ansatz $\rho := \rho_{\rm A}\otimes (\sx_n+\mu \sy_n)$, with a real $\mu$, we get $\cL_{\rm A}(\rho_{\rm A})\otimes (\sx_n+\mu \sy_n)+\{\rho_{\rm A},\sz_{n-1}\}\otimes (\sy_n-\mu\sx_n)=0$. Therefore, $\rho_{\rm A}$ must simultaneously satisfy $\cL_{\rm A} (\rho_{\rm A})=0$ and the zero anticommutator, $\{\rho_{\rm A},\sz_{n-1}\}=0$. Expanding $\rho_{\rm A}$ into an orthogonal basis $b^{(k)}$ acting on sites $1,\ldots,n-2$, $\rho_{\rm A}=\sum_k b^{(k)} \otimes r^{(k)}_{n-1}$, each $r^{(k)}_{n-1}$ must anticommute with $\sz_{n-1}$, and, therefore, must be from a linear span of $\{\sx_{n-1},\sy_{n-1}\}$, leading to ${\rm tr}_{\rm A} \rho_{\rm A}=0$. The reduced SS on the $n$-th spin is never off-diagonal.
\end{proof}
Note that, while sometimes a solution of the two conditions on $\rho_{\rm A}$ might not exist, there are cases where a traceless solution does exist~\cite{foot3}. A simple consequence of the above theorem is that, for the Ising-type Hamiltonian, $H=\sum_{j=1}^{n-1} \sz_j \sz_{j+1}$, and an arbitrary Lindblad Liouvillian on the first $n-1$ spins, on the last spin one is able to reach only all diagonal reduced SSs (the stabilizable set is only the $z$ axis of the Bloch ball). However, as we will now show, the Ising-type Hamiltonian is, in a sense, the worst choice, with other $H$'s being better~\cite{foot2}. We shall demonstrate this with a simple $2$-qubit example which is analytically solvable.

{\em Two-qubit systems.--}
Let us study the set of stabilizable reduced states for Hamiltonians of the form $H=\sx_1 \sx_2+\sy_1 \sy_2+\Delta \sz_1 \sz_2$. Expressing the reduced SS $\rho_2={\rm tr}_1(\rho_\infty)=\frac{1}{2}\1+\mathbf{r}\cdot\boldsymbol{\sigma}$ in terms of the Bloch vector $\mathbf{r}$, we already know that the set of reachable $\mathbf{r}$'s is equal to the whole Bloch ball for $\Delta=1$ (the isotropic Heisenberg model), while it is equal to a line $\mathbf{r}=(0,0,s), s \in [-1/2,1/2]$ for $\Delta \to \infty$ (the Ising model). We are now going to demonstrate that for any finite $\Delta$ the whole Bloch ball is reachable. 
\begin{figure}[t!]
\centering \includegraphics[width=1.1in]{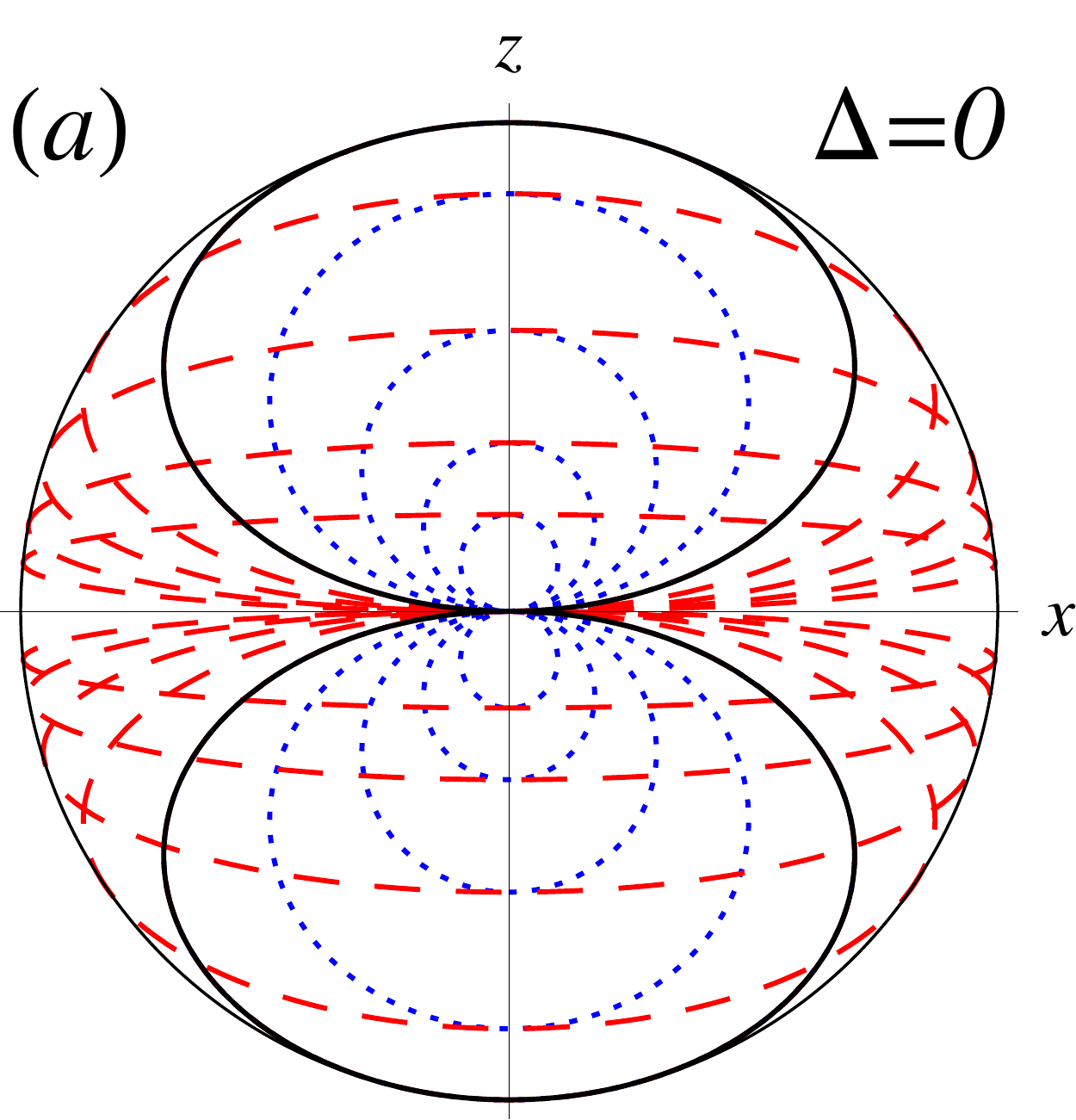} \includegraphics[width=1.1in]{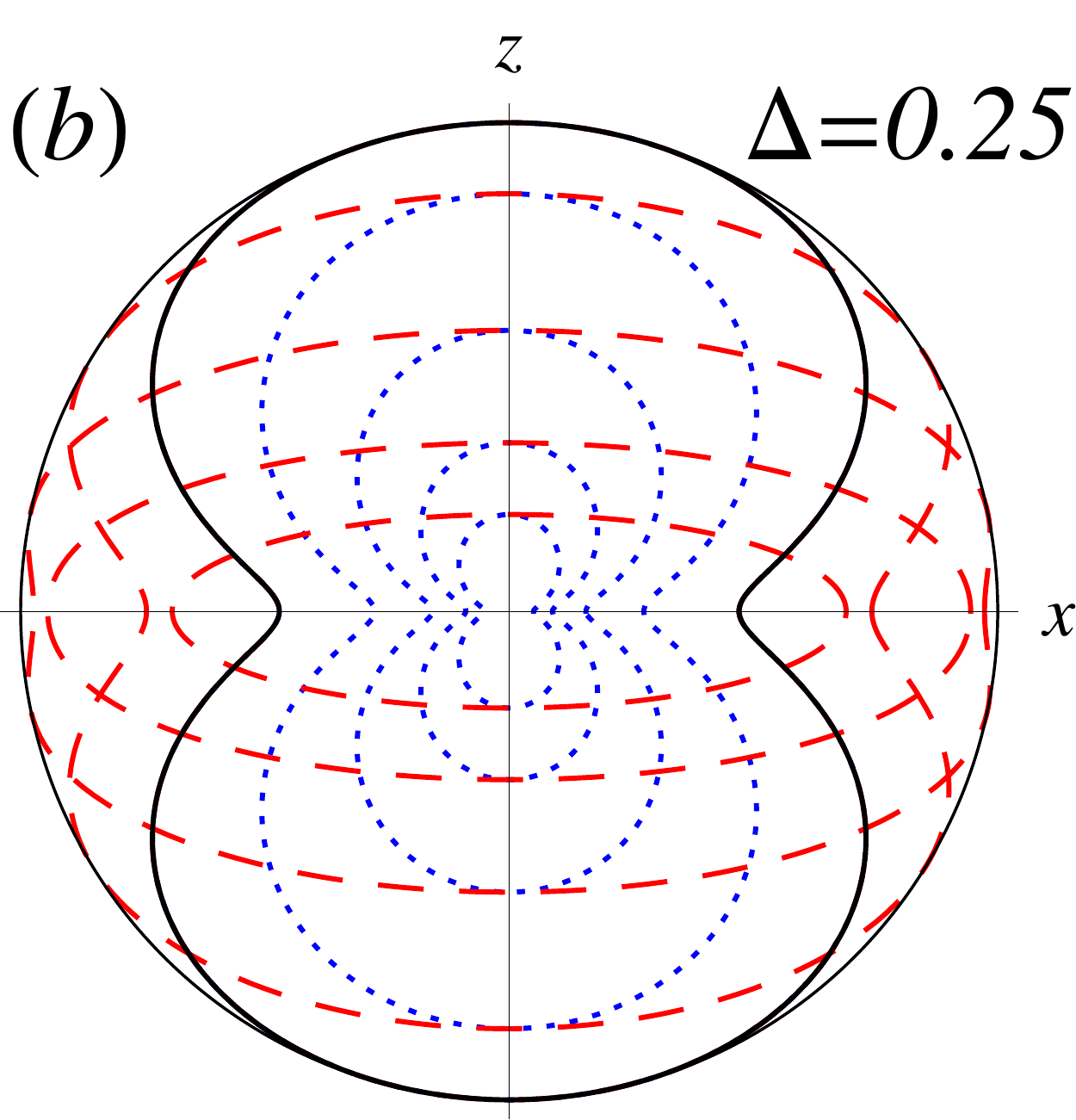} \includegraphics[width=1.1in]{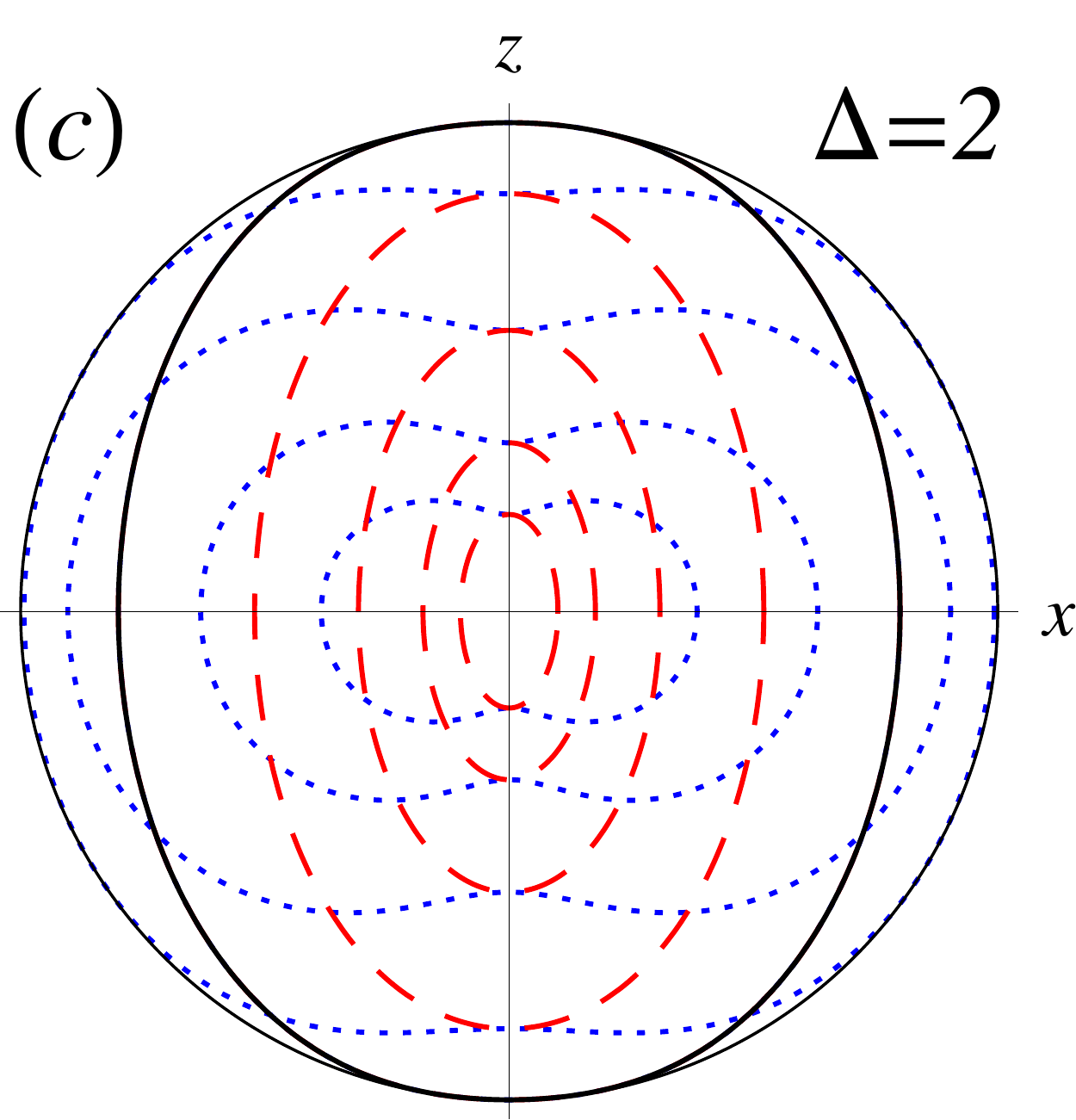}
\caption{(Color online) The Bloch vector $\mathbf{r}$ of reduced stabilizable states for an infinite strength $\cL_1$ and XXZ coupling on two spins, see the text. Different curves correspond to Eq.~(\ref{eq:r2}) for $k=1$ (black), $k=10^u,\quad u=0.5,1,1.5,2$ (red/dashed, decreasing size for increasing $k$), and $k=10^u,\quad u=-2,-1.5,-1,-0.5$ (blue/dotted, decreasing size for decreasing $k$). Varying $k$, all states can be reached.}
\label{fig:infiniteq}
\end{figure}
Let us take $\cL_1$ with a single Lindblad operator $L=\frac{1}{2\sqrt{2}}(\sqrt{q_3} (\sz_1 \cos{\beta}-\sx_1 \sin{\beta}) -\ii \sqrt{q_2}\sy_1)$. It is a rotated deformed $\sigma^+_1$, with the diagonal form parameters~\cite{geometry15} being $q_1=t_{2,3}=0$, while $t_1=\sqrt{q_2 q_3}$. For a given $q_{2,3}$ the chosen $t_1$ is the largest possible, resulting in the largest $\mathbf{r}$~\cite{geometry15}. The SS of such an $\cL$ can be computed explicitly, giving us the reduced SS $\rho_2$. The expression for $\mathbf{r}$ is still fairly complicated and we do not write it out. We notice that, provided $q_{2,3}$ are finite, not all states within the Bloch ball can be reached. Focusing on the limit in which we allow an $\cL_1$ of any strength, we set $q_2=q_3/k$ and take the limit $q_3 \to \infty$, in which the expression for $\mathbf{r}$ simplifies. Taking into account rotational invariance around the $z$ axis, we can limit our discussion to $\mathbf{r}$ laying in the $xz$ plane, writing $\mathbf{r}=r(\cos{\varphi},0,\sin{\varphi})$, and obtaining
\begin{equation}
r^2=\frac{k(\Delta^2+\tan^2{\varphi})(1+\tan^2{\varphi})}{[k(\Delta^2+\tan^2{\varphi})+1+\tan^2{\varphi}]^2}.
\label{eq:r2}
\end{equation}
We plot these curves for a set of $k$'s in Fig.~\ref{fig:infiniteq} . We see that varying $q_3/q_2=k$, the whole Bloch ball can be reached, except for $\Delta=0$, where we cannot reach $\mathbf{r}$ in the $xy$ plane (but can come arbitrarily close). A stabilizable set of reduced states in an important case of longer ${\rm XXZ}$-type chains, which is likely not analytically tractable, needs to be studied in future.

{\em Conclusion.--} We demonstrate that in the presence of a Heisenberg-type interaction one can prepare an arbitrary target one-site state at a distant remote location by acting with Markovian dissipation only on a single site. No control over the medium is required. We also study the convergence time of such a remote-state preparation procedure, finding that the fidelity has a universal scaling form and that, interestingly, the convergence time grows with a distance slower than suggested by the inverse gap of the propagator. We also characterize the set of reachable reduced SSs in the presence of other types of interaction, like the anisotropic Heisenberg coupling. We show that with the Ising interaction one can prepare only diagonal states, while with others (on two qubits) the stabilizable set is equal to the whole Bloch ball.


\begin{thebibliography}{99}
\bibitem{footSwap} Depending on the usage, an additional ${\rm SWAP}$ might be required at the target site to take the state ``out''.
\bibitem{foot4} Approximately the same scaling exponents are obtained also for other choices of generic initial states, like e.g., random product states, or an infinite temperature mixed state $\rho(0)=\frac{1}{2^n}\mathbbm{1}$ (data not shown). For those two $\nu$ is closer to $1.0$ and $3.0$ for $\rho_{n/2}$ and $\rho_n$, respectively.
\bibitem{foot2} To allow for off-diagonal reduced SSs it is enough to add a magnetic field on the last site. That is, taking $\cL_{\rm A}\otimes \1_n$ and $H=(\1+\sz)_{n-1}\otimes \sz_n$ one can get solutions $\rho_{\rm A}$ with nonzero trace. A two-qubit example is $L=\sz_1$ and $H=\sz_1\sz_2+\sz_2$, for which the steady-state subspace is spanned by $\{ \1,\sz_1,\sz_2,\sz_1\sz_2,(\1-\sz)_1\sx_2,(\1-\sz)_1\sy_2 \}$, allowing in turn for any one-qubit SS on the second spin.
\bibitem{foot3} A simple example on two spins is obtained by taking one $L=\sx_1$ and $H=\sz_1\sz_2$, resulting in the SS subspace spanned by $\{\1_1\1_2,\1_1\sz_2,\sx_1\sx_2,\sx_1\sy_2\}$. This is an example of a non-diagonal SS~\cite{Buca}.
\end{thebibliography}
\end{document}